\documentclass[11pt]{article}
\usepackage[utf8]{inputenc}
\usepackage[T1]{fontenc}
\usepackage[english]{babel}
\usepackage{amsmath,amssymb}
\usepackage{color}
\usepackage{pdfsync}    

\usepackage{verbatim}
\usepackage{amsfonts}
\usepackage[pdftex]{graphicx}
\usepackage{amsmath,amsfonts,amssymb,theorem,euscript,array,enumerate,mathrsfs,cases,bbm}

\usepackage{geometry}
\usepackage{eurosym}
\usepackage{hyperref}
\hypersetup{
    colorlinks,
    citecolor=blue,
    filecolor=blue,
    linkcolor=blue,
    urlcolor=blue
}
\usepackage{todonotes}
\usepackage{mathrsfs}

 \usepackage{rotating}
\usepackage[round]{natbib}
\setcitestyle{round}

\def\be{\begin{eqnarray}}
\def\ee{\end{eqnarray}}

\def\b*{\begin{eqnarray*}}
\def\e*{\end{eqnarray*}}

\newtheorem{Theorem}{Theorem}[part]

\newtheorem{Proposition}{Proposition}[part]

\newtheorem{Lemma}{Lemma}[part]

\newcommand{\ba}{\begin{array}}
\newcommand{\ea}{\end{array}}
\newcommand{\ben}{\begin{equation*}}
\newcommand{\een}{\end{equation*}}
\newcommand{\bea}{\begin{eqnarray}}
 \newcommand{\eea}{\end{eqnarray}}
\newcommand{\bean}{\begin{eqnarray*}}
\newcommand{\eean}{\end{eqnarray*}}
\newcommand{\bel}{\begin{align}}
\newcommand{\eel}{\end{align}}
\newcommand{\beln}{\begin{align*}}
\newcommand{\eeln}{\end{align*}}
\newcommand{\bit}{\begin{itemize}}
\newcommand{\eit}{\end{itemize}}

\makeatletter \@addtoreset{equation}{section}

\@addtoreset{Definition}{section}

\@addtoreset{Theorem}{section}

\@addtoreset{Proposition}{section}

\@addtoreset{Property}{section}

\@addtoreset{Assumption}{section}

\@addtoreset{Corollary}{section}

\@addtoreset{Lemma}{section}

\@addtoreset{Remark}{section}

\@addtoreset{Example}{section}

\@addtoreset{Condition}{section}

\newcommand{\pink}[1]{\textcolor[rgb]{1.00,0.00,0.50}{#1}}

\def \E{\mathbb{E}}

\def \H{\mathbb{H}}
\def \L{\mathbb{L}}
\def \M{\mathbb{M}}
\def \N{\mathbb{N}}
\def \P{\mathbb{P}}
\def \Q{\mathbb{Q}}
\def \R{\mathbb{R}}

\def \Z{\mathbb{Z}}

\def \G{\mathbb{G}}
\def \qi ={ \tilde q_i}

\def\Fc{{\cal F}}

\newcommand{\bA}{{\mathbf{A}}}

\def\={\;=\;}
\def\.{\;.}

\def\b{\beta}

\def\1{{\bf 1}}

 \def\normeL2#1{\left\|{#1}\right\|_{L^2}}

\DeclareMathOperator{\Corr}{Corr}

\newcommand{\alias}[2]{
\providecommand{#1}{}
\renewcommand{#1}{#2}
}

\alias{\P}{\mathbb{P}}
\alias{\N}{\mathcal{N}}
\alias{\L}{\mathcal{L}}
\alias{\Z}{\mathbb{Z}}
\alias{\Q}{\mathbb{Q}}
\alias{\R}{\mathbb{R}}
\alias{\C}{\mathcal{C}}
\alias{\T}{\mathbb{T}}
\alias{\E}{\mathbb{E}}
\alias{\H}{\mathcal{H}}
\alias{\1}{\mathbf{1}}
\alias{\B}{\mathcal{B}}
\alias{\M}{\mathcal{M}}
\alias{\G}{\mathcal{G}}
\alias{\Y}{Y_{\bullet}}

\newcommand{\nc}{\newcommand}
\nc{\cA}{{Y}} \nc{\cB}{{\mathcal B}} \nc{\cC}{{\mathcal
C}} \nc{\cD}{{\mathcal D}} \nc{\bbD}{\mathbb{D}}
\nc{\cG}{{\mathcal G}} \nc{\cF}{{\mathcal F}} \nc{\cS}{{\mathcal
S}} \nc{\cU}{{\mathcal U}} \nc{\cH}{{\mathcal H}}
\nc{\cK}{{\mathcal K}} \nc{\cM}{{\mathcal M}} \nc{\cO}{{\mathcal
O}} \nc{\cP}{{\mathcal P}} \nc{\bbE}{\mathbb{E}}
\nc{\bbEP}{\mathbb{E}_{\mathbb{P}}}\nc{\bbL}{\mathbb{L}}
\nc{\bbP}{\mathbb{P}} \nc{\bbQ}{\mathbb{Q}} \nc{\del}{\partial}
\nc{\Om}{\Omega} \nc{\om}{\omega} \nc{\bbR}{\mathbb{R}}
\nc{\bbC}{\mathbb{C}} 
\nc{\bfr}{\begin{flushright}}
\nc{\efr}{\end{flushright}} 
\nc{\dXt}{\Delta X_{t}}
\nc{\dXs}{\Delta X_{s}} \nc{\bs}{\blacksquare} \nc{\dX}{\Delta X}
\nc{\dY}{\Delta Y}
\nc{\dnkx}{\left(X(T^{n}_{k})-X(T^{n}_{k-1})\right)}
\nc{\esssup}{\mathrm{ess}\mbox{ }\mathrm{sup}}
\nc{\essinf}{\mathrm{ess}\mbox{ } \mathrm{inf}}
\nc{\dhats}{\widehat{\delta_s}}

\nc{\chf}{\mbox{$\mathbf1$}}
\nc{\ind}{\mathds{1}}

\nc{\hs}{\vspace{3mm}}

\begingroup\expandafter\expandafter\expandafter\endgroup
\expandafter\ifx\csname pdfsuppresswarningpagegroup\endcsname\relax
\else
\fi
\graphicspath{{./figs/}}

\newcommand{\mg}{\textcolor[rgb]{0.00,0.59,0.00}}

\begin{document}

\title{Emission impossible: \\ 
 Balancing  Environmental Concerns and Inflation\thanks{We thank Marc Baudry, Giorgia Callegaro, Estelle Cantillon, Anna Creti, Giorgio Ferrari, Fausto Gozzi Frank Riedel, Karin Mayr-Dorn, Aurelie Slechten and conference and seminar participates at the AMaMeF, FSR Climate Annual Conference, Spring Colloquium on Probability and Finance, Economics of Climate Change Seminar. The authors' email addresses are \href{mailto:rene.aid@dauphine.psl.eu}{\texttt{rene.aid@dauphine.psl.eu}}, 
 \href{mailto:mariarduca@gmail.com}{\texttt{mariarduca@gmail.com}}, 
 \href{mailto:sbiagini@luiss.it}{\texttt{sbiagini@luiss.it}} and 
\href{mailto:Luca.Taschini@ed.ac.uk}{\texttt{luca.taschini@ed.ac.uk}}.}}
\author{René Aïd\quad Maria Arduca\quad Sara Biagini \quad Luca Taschini}

\maketitle

\begin{abstract}
We provide a theoretical framework to examine how carbon pricing policies influence inflation and to estimate the policy-driven impact on goods prices from achieving net-zero emissions. Firms control emissions by adjusting production, abating, or purchasing permits, and these strategies determine emissions reductions that affect the consumer price index. We first examine an emissions-regulated economy, solving the market equilibrium under any dynamic allocation of allowances set by the regulator. Next, we analyze a regulator balancing emission reduction and inflation targets, identifying the optimal allocation when accounting for both environmental and inflationary concerns. By adjusting penalties for deviations from these targets, we demonstrate how regulatory priorities shape equilibrium outcomes. Under reasonable model parameterisation, even when considerable emphasis is placed on maintaining inflation at acceptable levels or grant lower priority to emissions reduction targets, the costs associated with emission deviations still exceed any savings from marginally lower inflation. Emission reduction goals should remain the primary focus for policymakers.
\end{abstract}

\vspace{0.15cm}

\vspace{0.15cm}

\vspace{0.15cm}
\noindent{\bf Acknowledgements:} We would like to thank the Finance for Energy Markets Research Initiative, the Chair EDF–Credit Agricole CIB  \textit{Finance and Sustainable Development: Quantitative Approach,} the French ANR PEPR Math-Vives project MIRTE ANR-23-EXMA-0011, and the LUISS visiting program for their financial support.


\clearpage
\section{Introduction}\label{sec:intro}
\setcounter{page}{1}

Carbon pricing has emerged as a key policy instrument in the global effort to mitigate climate change, with a growing number of jurisdictions adopting carbon taxes or cap-and-trade systems to curb greenhouse gas emissions (\cite{WorldBank2024}). 
These policies, which internalize the environmental cost of carbon-intensive activities, have been shown to be effective in reducing emissions (\cite{Martin2016}, \cite{Dechezlepretre2023}, among others). 
However, this effectiveness does not come without economic costs.
By altering relative prices, carbon pricing raises the cost of emission-intensive goods, potentially affecting the overall price level. 
The empirical literature on the inflationary impact of carbon pricing offers mixed evidence. Some studies find that the impact on inflation has been relatively weak (\cite{Metcalf2019}; \cite{Metcalf23}; \cite{Konradt23}) while others indicate more substantial inflationary effects (\cite{Kanzig23}).
Criticism of carbon pricing and promises to ``axe the tax'' have surfaced in political debates, which may in turn prompt policymakers to slow down, weaken, or reconsider the scope of climate policies.
This growing political and public scrutiny of carbon pricing has sparked discussions about the opportunity to balance inflation concerns with environmental targets to ensure that these policies remain both effective and publicly acceptable.

In this paper, we develop a theoretical framework to explore how policymakers can trade off less ambitious carbon emission reduction targets for lower inflationary pressures arising from carbon pricing.
We consider a regulatory authority that determines the optimal allocation of emission allowances by explicitly balancing its tolerance for inflationary pressure with the need to meet stringent emission reduction goals. 
Pressing too hard on emissions reductions risks driving up inflation, whereas being overly cautious on inflation may impede progress in decarbonizing the economy. 
We derive the optimal allowance allocation by factoring in both environmental and inflationary concerns. By adjusting penalties for deviations from these targets, we demonstrate how regulatory priorities shape equilibrium outcomes.
Whether inflation concerns should influence policy as a means to secure public acceptance is ultimately an empirical question. 
Calibrating the model to data, our results indicate that emissions reduction goals should remain the primary focus. Under reasonable parameterization, even when policymakers devote considerable effort to maintaining inflation at acceptable levels or assign lower priority to emissions reduction targets, the costs tied to emission deviations still outweigh any benefits from marginally lower inflation.

We contribute to the debate on the effects of carbon pricing on inflation by presenting a theoretical framework to study how carbon pricing policies influence inflation and by applying this model to calculate the theoretical policy-driven impact on goods prices resulting from achieving net-zero emissions.
Our analysis considers an economy where emissions are limited by a regulatory authority through a cap-and-trade system.
Firms can employ three strategies to comply with the policy: adjusting production levels, undertaking costly pollution abatement, or purchasing permits. 
These strategies determine the total amount of controlled emissions reductions, which, in turn, influence changes in the consumer price index.
We begin by solving the compliance problem, deriving the firms’ optimal compliance strategies. 
Next, we solve for the market equilibrium, obtaining the equilibrium permit price given a specific emissions cap and corresponding allocation program. 
Essentially, we model the primary transmission channel for policy-driven inflation: changes in production. These production adjustments are influenced by firms’ strategic decisions between abatement and trading, which, in turn, depend on the permit price set in the cap-and-trade market. The permit price ultimately reflects the availability of permits, determined by the allocation program established by the regulator.
We then consider an alternative policy approach where the regulator decides the allocation program not only to meet emission reduction targets but also to account for policy-driven inflationary pressures. 
This introduces a deliberate trade-off in the regulatory approach, balancing environmental objectives with inflationary pressure.
To formalize this trade-off, we explicitly model it within the regulator’s objective function. The regulator aims to minimize the overall compliance costs incurred by firms, while simultaneously managing deviations from both a pre-set emission reduction target and an acceptable inflation target. 
We demonstrate how the equilibrium permit price is inherently tied to this trade-off and depends on the regulator’s priorities. If the regulator places greater emphasis on achieving stringent emission reduction targets, the allocation program will be designed to impose a tighter cap, leading to higher permit prices. Conversely, if the regulator is more focused on minimizing inflationary impacts, the allocation program may allow for more permits, thereby reducing the upward pressure on permit prices and the resulting inflation.

We apply the theoretical model to empirical data, using an illustrative example as a unifying thread throughout the paper.
First, we illustrate how the allocation program affects controlled emissions, which are determined by firms' optimal decisions to abate emissions, adjust production, or trade permits. These decisions are directly influenced by the stringency of the allocation program and, in turn, determine the equilibrium permit price.
Second, we continue the example to calculate the theoretical inflationary effects of achieving net-zero emissions in Europe by 2050, assuming no additional consideration of inflationary impacts. This step provides a baseline for understanding the potential price-level implications of ambitious carbon reduction targets.
Third, we evaluate the trade-off between emissions reduction and inflationary impacts, focusing on the regulator’s decision-making process. By varying the penalization parameters associated with deviations from emission reduction and inflation targets, we illustrate how the regulator balances these competing objectives. 
This analysis emphasizes how the relative priority given to inflation versus emissions targets shapes equilibrium outcomes, particularly the permit price and overall policy costs.
The empirical results reveal that, under reasonable model parameterisation, even when significant weight is placed on maintaining inflation near acceptable levels or when the emphasis on achieving emissions reduction targets is substantially relaxed, emission-related costs remain a central focus of the regulator’s objectives. 
This underscores the importance of ensuring that emissions reduction goals continue to outweigh concerns over inflation deviations, reflecting the primary intent of cap-and-trade systems to drive meaningful emission reductions.

While carbon pricing policies may lead to higher prices for carbon-intensive goods, they are also expected to foster innovation in low-carbon technologies, thereby reducing the demand for carbon-intensive inputs, abatement measures, and permits. These countervailing forces could mitigate, or even counterbalance, the inflationary effects of carbon pricing.
Although our model does not explicitly account for investments in new production technologies, we acknowledge that incorporating such dynamics would likely further reduce the inflationary impact of carbon pricing, thereby reinforcing the robustness of our policy conclusions.
The economic impact of carbon pricing has also been shown to depend on the level of coverage of the initiatives. Broader coverage, encompassing a larger proportion of emissions within a jurisdiction, tends to result in greater pass-through to consumer prices (\cite{Metcalf-stock}). 
We recognize that the effects of carbon pricing on inflation can depend on policy coverage and additional factors, such as the production structure of the economy and the extent to which a country specializes in high-polluting industries and sectors. While these aspects are undeniably important, we believe that the general policy conclusions derived from our analysis remain applicable to such contexts without requiring modifications to the model.

Carbon pricing programs are often accompanied by measures to enhance public acceptability. Revenue recycling --whether through direct redistribution to households, reductions in other taxes, or investments in low-carbon infrastructure-- has been shown to mitigate some of the economic costs of carbon pricing and alleviate its regressive impacts on lower-income groups (\cite{Carattini2018} and \cite{Klenert2018}). While revenue recycling is an effective tool for addressing equity and to some extent inflation concerns, our analysis focuses on an alternative scenario where redistribution is either not feasible or is implemented on a limited scale. This could be due to institutional constraints, political challenges, or competing priorities for the allocation of carbon revenues.

\section{The Model} \label{sec:model}

We model the interconnection between production activity and pollution control strategies in a dynamic set up. 
We represent an economy where firms are required to operate under a cap-and-trade system where a cap is imposed on the total amount of greenhouse gas (GHG) emissions allowed within the economy.\footnote{Hereafter, we use GHG emissions and carbon emissions interchangeably.}
Under cap-and-trade, firms generally use a combination of three strategies to comply with GHG emission limits: adjusting production levels, costly pollution emission abatement, and purchasing permits.\footnote{When the cost of purchasing permits is fixed, this scenario mirrors an economy operating under a carbon tax system. In this case, firms can choose to reduce emissions by adjusting production levels, investing in abatement, or simply paying the tax on their emissions.}
The production strategy involves firms maintaining their regular production while keeping emissions within prescribed limits by adjusting their output levels.
Alternatively, firms can increase abatement measures to reduce emissions without significantly altering their output or operations. This may involve improving energy efficiency, expanding pollution control equipment, or scaling up carbon capture and storage technologies.
The third approach involves offsetting emissions by using permits to emit. In a cap-and-trade system, firms receive emission permits and can buy additional permits if they exceed their allocation or sell excess permits if they emit less.
These strategies offer firms flexibility in complying with pollution regulations, allowing them to combine the three strategies to optimize their compliance efforts and meet GHG targets.
The selection of a compliance strategy is fundamentally driven by its relative cost-effectiveness.\footnote{Cost-effectiveness means achieving a given reduction at minimum cost or maximizing GHG emission reductions at a given cost. If the costs of abatement measures or purchasing permits exceed the costs associated with adjusting production, firms may opt to reduce output to meet pollution restrictions.}$^,$\footnote{Later we will show how the trade-off between emission reduction and inflation depends on firm responses to the policy, e.g. production reduction or abatement or purchasing permits, and how these translate into consumer price changes.}
Firms will weigh the costs associated with each option --whether adjusting production levels, abating emissions, or purchasing permits-- against the potential financial impacts of non-compliance with pollution restrictions.

\vspace{0.25cm}
We consider an economy with $N$ firms, all of which are subject to a cap-and-trade system over a compliance period of length $T$.
Each firm must cover its unabated emissions with permits or face a penalty (\cite{Can-Slech} and \cite{Carmona-Hinz}).
We work on  a filtered probability space  $(\Omega, {(\Fc_t)}_{0\leq t\leq T}, P)$  generated by an $N+1$ dimensional Brownian motion $(\tilde{W}^0, \tilde{W}^1, \ldots, \tilde{W}^N)$.
The component $\tilde{W}^0$ represents the common noise factor affecting all firms, linked to broader economic conditions, such as the general business cycle, which influences aggregate product demand across the economy. 
In contrast, the terms $\tilde W^{i}, i=1, \ldots N$ represent idiosyncratic noise specific to each firm. These components capture the individual variations in product demand and production that arise within a firm’s specific sector or business operations, reflecting the unique uncertainty faced by each firm independently of the broader economy.
We define all relevant processes, including strategies and prices, as square-integrable adapted processes. 
These processes belong to the space $\L^2(\Omega\times [0,T]),$ which we refer to simply as $\L^2$:
$$\L^2 =\left  \{ \phi \Big\vert \phi \text{ adapted and }  \E \left[\int_0^T \phi^2(t) dt\right ]< \infty\right \} $$
Also,    $(\L^2)^N = \L^2 \times \ldots \times \L^2$, $  N\geq 2 $ denots the $N$-dimensional version of $\L^2$. 
The terminal positions belong to the space of square-integrable random variables, denoted by $L^2(\Omega,\cF_T,\P)$, or simply $L^2$ for shorthand. Additionally, $L^2_+$ represents the positive cone within the space $L^2$.

\paragraph{Production, abatement, and trade --} The firm's production process generates carbon emissions.
Letting $q^i_t$ represent the production rate of the $i-$th firm, the dynamics of the \textit{controlled cumulative carbon emissions} (hereafter referred to simply as controlled emissions) associated with the production of goods by firm $i-$th are described as follows:
\begin{align}
   dE^{i}_t& =(\gamma_i  q^i_t - \alpha^i_t )\, dt + \sigma_i dW^i_t, \quad \quad \gamma_i >0, \sigma_i > 0. \label{eqn:emissions}
\end{align}  
The parameters $\gamma_i$ and $\sigma_i$ represent the emission intensity per unit of production and the volatility of the firm's emissions, respectively.
While firms can influence the trend or average behavior of their emissions through adjusting production rates $q^i$ and selecting an abatement rate $\alpha^i,$ the variability of emissions, $\sigma_i,$ remains beyond their direct control and is constant.
This uncontrollable component reflects real-world factors like fluctuating product demand, operational inefficiencies, or environmental conditions that introduce variability in emissions. Incorporating this aspect into the model captures inherent uncertainties in pollution emission control. 
The trend $(\gamma_i  q^i_t - \alpha^i_t )$ is not restricted in sign, allowing firms to act as either net polluters or net cleansers, depending on whether their emissions exceed or are offset by abatement efforts. 
The shocks in firm $i$'s emissions follow a Brownian motion defined as:
$$W^i:= s_i \tilde W^0 + \sqrt{1-s_i^2}\tilde W^i,$$
where $s_i \in [-1,1]$.  
This formulation accommodates a general correlation structure between firms' shocks, with the correlation given by $\Corr (W^i,W^j) =\rho_{ij} = s_is_j$. 

\vspace{0.15cm}
The firm's decision to adjust its production rate affects output prices.
We assume the firm faces the following inverse demand function for its goods:
$$ S_i(q) := a_i - b_i q, \quad \quad a_i >0, b_i>0$$
where $a_i$ is the base price and $b_i$ is the rate at which the price decreases with each additional unit produced. 
The associated cost function is:
$$c_i(q) := \kappa_i (q-\tilde q_i ) + \frac1{2} \delta_i (q-\tilde q_i )^2, \quad \quad 0< \kappa_i < a_i , \, \delta_i>0$$
where $\tilde q_i$ represents the cost-minimizing solution in the production problem without regulation, specifically in the laissez-faire scenario, which we discuss later.

\vspace{0.15cm}
The Expression \eqref{eqn:emissions} captures the idea that, much like in real-world scenarios, firms can manage their carbon emissions by implementing abatement projects --- like enhancing energy efficiency, reducing fugitive emissions, or increasing recycling processes.
The cost of these projects is represented by the following linear-quadratic model:
$$ g_i(\alpha) := h_i \alpha + \frac1{2\eta_i} \alpha^2, \quad \quad h_i>0, \eta_i>0,$$
where the parameters $h_i$ and $\eta_i$ control the intercept and slope of $i$’s linear marginal abatement cost, respectively.\footnote{This analytical framework is widely used in emissions management and abatement cost modeling, as discussed in reviews such as \cite{Salant2015} and \cite{Can-Slech}} 
The parameter $\eta_i$ also indicates the reversibility of the abatement decision, with higher values suggesting greater flexibility and a higher likelihood that the decision can be reversed.

\vspace{0.15cm}
Under a cap-and-trade system, each firm is allocated a certain number of permits, representing their allowed cumulative emissions for the regulated period $T.$
Firms can trade these permits among themselves: if a firm’s emissions exceed its allocation, it can purchase additional permits from other firms$(\beta^i >0);$ if it emits less than its allocation, it can sell the excess permits $(\beta^i <0)$ or bank them for future use.
Banking unused permits creates a reserve for future compliance, offering flexibility and more efficient emissions management, especially when anticipating future increases in abatement costs or changes in permit allocations.
The dynamics of the permit bank are described as follows:
\begin{equation} \label{eq: emiss-bank}
dX^i_t = -dE^i_t + \beta^i_t dt + dA^i_t, \quad X^i_0 =A^i_0, 
\end{equation}
where the process $A^i$ denotes the cumulative allocation (supply) of permits from the start of the period until time $t$, while $\beta^i$ represents the rate of permit trading.
We will explore the optimal allocation program $\textbf{A}:=(A^1, \ldots , A^N$) in more detail in Section \ref{sec:regulation}, where we study the regulator's problem.   


\paragraph{Laissez-faire benchmark --} \label{par: laissez} To isolate the impact of emission pollution control on a firm's decision-making, we first examine a laissez-faire environment where there are no limits on emissions. 
In this scenario, each firm operates as a monopolist in its market for the $i-$th good and determines the optimal production level, $\tilde q_i,$ without considering emissions:
\begin{align*}
\max_{q } \{S_i(q ) q  - c_i(q )\} =\max_{q } \left\{ (a_i - b_i q)q  - \kappa_i (q-\tilde q_i ) - \frac1{2} \delta_i (q-\tilde q_i)^2\right\}
\end{align*}
We use a standard quadratic cost function, which implies that reducing production --and thereby emissions-- lowers marginal production costs. 
By solving for the first-order condition, we determine the optimal production level:
$$
q^* := \frac{a_i -
\kappa_i+\delta_i \tilde q_i}{\delta_i + 2b_i}
$$
Equating this to the expected optimal quantity $\tilde q_i,$ gives us: 
\begin{align}\label{eq: q0}
\tilde q_i = \frac{a_i - \kappa_i}{2 b_i}.
\end{align}
The parameters in the production cost function are chosen to ensure that $\tilde q_i$ is strictly positive. 
Notably, without regulation, the firm's emission intensity does not influence the optimal production level.

\paragraph{Emission-controlled economy --} We now examine the case where firms must comply with emission pollution control measures. 
Firms incur costs related to emissions and manage them through production adjustments, abatement efforts, and permit trading.
Each firm optimizes the balance of production production $\hat{q},$ abatement $\hat{\alpha}$, and trading $\hat \beta$ strategies to minimize total costs:  
\begin{align} \label{eq:objfirm}
\inf_{q^i,\alpha^i,\beta^i} J^i(q^i,\alpha^i,\beta^i) := \E\left[ \int_0^{T} \Big( -S_i(q^i_t) q^i_t + c_i(q^i_t) + g_i(\alpha^i_t) +  \beta^i_t P_t \Big) dt + \lambda( X^i_T )^2 \right],
\end{align}
where $P$ is the market price of emission permits, and we assume a frictionless permit market.\footnote{For models that account for permits trade friction, see \cite{quemin-baudry} and \cite{Aid-Bia}}
Each firm also faces costs related to over- or under-compliance with pollution regulations, represented by the final term in Expression (\ref{eq:objfirm}). To manage this, a penalty is applied to the cumulative net emissions position over the period $[0,T],$ accounting for emissions, abatement, permit trading, and the initial allocation of permits. 
This penalty applies to the terminal permit balance $X^i_T = E^i_T - \int^T_0 \beta_t^i dt - A_T^i,$ with $E_T^i$ and $A_T^i$ as the cumulative emissions and permit allocations. 
A penalty factor $\lambda$ is imposed for each unit of excess emissions, ensuring firms have a financial disincentive for exceeding their permit limits.\footnote{Compliance is typically enforced by assuming perfect adherence, as in \cite{Rubin}, \cite{Kol-Tas}, and papers reviewed in \cite{Salant2015}. Alternatively, compliance can be encouraged by imposing a substantial penalty for excess emissions, as in \cite{Carmon-Fehr-Hinz}, \cite{Carmona-Hinz}, and studies summarized in \cite{Can-Slech}.}$^,$\footnote{The penalty is quadratic to account for both under-compliance and excess compliance, as the latter would represent wasted resources and unnecessary costs.}


\section{Market equilibrium}

\subsection{Firms optimal strategies}
We now solve the cost-minimization problem for an individual firm, assuming the permit price~$P$ and firm individual allocation program $A^i$ are given.
In the next section, $P$ is determined as the market equilibrium price for a given allocation $\mathbf{A},$ and the optimal allocation program is explored in Section \ref{sec:regulation}. 

\hs
Consider the process:
\[
M^i_t:=\E_t[A^i_T],
\]
where $\E_t[A^i_T]$ represents the conditional expectation at time $t$ of the firm's cumulative permit allocation $A^i_T$ over the period $[0,T]$. 
Here, $M^i$ is the martingale generated by the cumulative allocation $A^i_T.$
Importantly, different intertemporal allocations that result in the same cumulative total allocation over $[0,T]$ will generate the same $M^i.$
In other words, if different allocation schedules provide the same total number of permits over the regulatory period, they will lead to the same $M^i.$
This means that the key factor influencing firms' compliance strategies is the total cumulative permits expected, rather than their specific timing or distribution.\footnote{This requires the 'no violation of availability' (\cite{SALANT20166}) or 'feasibility conditions' (\cite{Perino2016}) to hold. This ensures that optimal strategies remain unaffected by the specific allocation scheme. The timing of permit allocation only becomes important when firms have an incentive to stop banking permits and start borrowing from future allocations (\cite{SALANT20166}, \cite{Perino2016}, and \cite{Kol-Tas}).
When banking stops and borrowing is not allowed, the intertemporal problem breaks down. Firms then cover per-period regulated emissions through a mix of instantaneous abatement and trading, based on their relative costs.}

\vspace{0.15cm}
In the Business As Usual (BAU) scenario, we denote:
$$ \mu_i:=  \gamma_i \tilde q_i, \ \ \ \ \ \bar \mu_b:= \frac{1}{N}\sum_{i=1}^N \mu_i, $$
where $\mu_i$ is the emission drift of firm $i$ under BAU, and $\bar \mu_b$ represents the average emission trend across all $N$ firms or, equivalently, the average emissions responsibility per firm. 
We also denote:
$$  \psi_i:= \frac{\gamma_i^2}{\delta_i+2b_i}, \ \ \ \ \ \bar{\psi}:= \frac{1}{N}\sum_{i=1}^N \psi_i \   \text{ and } \ \bar H = \frac1N  \sum_{i=1}^N h_i\eta_i $$
The parameter $\bar\psi$ is the average intensity of emissions, adjusted for the combined effects of cost and market conditions on production -- we recall that $\gamma_i$ is the emission intensity per unit of production for $i-$th firm, $\delta_i$ is the sensitivity of the cost function to production changes, $b_i$ is the price sensitivity to output changes, and
$\bar H$ captures the average sensitivity of abatement costs in the economy.

\begin{Theorem}\label{Theo:single-firm-opt}
\textbf{Optimal compliance strategies} Given the cumulative allocation program $A^i$, the firm's problem in \eqref{eq:objfirm} has a solution if and only if the price process $P$ is a martingale. 
If this condition is satisfied, the optimal production and optimal abatement rates $(\hat{q}^i,\hat{\alpha}^i)$ in $(\L^2)^2  $ are unique:
\begin{align}
\label{eq: opt q}
&\hat{q}^i_t = \tilde q_i -\frac{\gamma_i}{\delta_i + 2b_i}P_t , \\
&\hat{\alpha}^i_t = \eta_i(P_t-h_i) ,
\label{eq: opt alpha}
\end{align}
An optimal trading rate is the solution to the following equation: 
\begin{eqnarray} \label{eq: opt beta} \nonumber
d \hat \beta^i_t& =&  -\left( \frac{1}{2\lambda(T-t)} + \eta_i +\psi_i\right) dP_t +
                \frac{1}{T-t}(\sigma_i dW^i_t - dM^i_t)\\
               &=& - d\hat \alpha_t^i-\left( \frac{1}{2\lambda(T-t)}   +\psi_i\right) dP_t+
                \frac{1}{T-t}(\sigma_i dW^i_t - dM^i_t)\\ \label{eq: opt beta0} 
\hat \beta_0 &  = &   \gamma_i \tilde q_i  - \hat \alpha^i_0   -\left( \frac{1}{2\lambda T} + \psi_i\right) P_0  
                - \frac{M^i_0}{T }  .      
\end{eqnarray}

\end{Theorem}

\textit{Proof:} The proof is provided in  the Appendix \ref{proof:Theo single-fim-opt}.
\\

\noindent We begin by examining the optimal production rate $\hat{q}$ in expression (\ref{eq: opt q}). 
When the permit price $P_t$ is near zero, a firm maintains its original production level since emissions restrictions impose no significant costs, allowing the firm to operate close to its laissez-faire rate without incurring substantial additional expenses.
However, as permit prices rise, firms are incentivized to reduce production.
The degree of production adjustment increases with the firm's emission intensity $\gamma_i$; the higher the emission intensity, the greater the reduction in production. 
Conversely, the adjustment in production decreases with higher values $\delta_i$ and $b_i.$
The denominator $(\delta_i + 2 b_i)$ captures the combined effect of cost and market conditions on production adjustment. 
A larger denominator reduces the quantity adjustment $\hat{q},$ reflecting the fact that both high cost sensitivity (large $\delta_i$) and high price sensitivity (large $b_i$) discourage significant production changes.\footnote{Economically, this denominator can be interpreted as a measure of ``resistance" to production adjustments under pollution regulation. When either the cost function is highly sensitive (large $\delta_i$) or the price drops rapidly with increased production (large $b_i$), the optimal response is to minimize production adjustments, leading to a production rate close to its laissez-faire solution.} 

Next, we consider the optimal abatement strategy $\hat{\alpha}$ in expression (\ref{eq: opt alpha}).
The parameter $h_i$ represents the fixed per-unit cost of abatement. Abatement only occurs when the permit price $P_t$ exceeds this cost threshold. 
The extent of abatement is directly proportional to the difference between the permit price and the abatement cost threshold, scaled by the abatement technology parameter $\eta_i$.
A higher $\eta_i$ indicates a more advanced and efficient abatement technology, leading to greater emission reductions for the same expenditure. 
Additionally, when $\eta_i$ is viewed as a measure of the reversibility of the abatement decision, higher values suggest greater flexibility in the abatement process, enhancing the firm's ability to reverse its abatement decisions if necessary.

Lastly, we examine the optimal trading rate, starting from $\hat \beta_0$ in Expression (\ref{eq: opt beta0}).
When the initial permit price is close to zero, the unconditional trading rate is approximately equal to the BAU emissions net of abatement, $(\gamma_i \tilde q_i  - \hat \alpha^i_0),$  minus the average expected per-period permit allocation, $({M^i_0}/{T}).$
As the permit price rises, firms are incentivized to increase abatement efforts and adjust their production rates, leading to a reduction in trading activity, as evidenced by the third term in Expression (\ref{eq: opt beta0}). 
The reduction in trading diminishes with higher penalty values and the length of the regulated period. 
The penalty $\lambda$ represents the opportunity cost of non-compliance, meaning that a higher penalty heightens the urgency for firms to secure permits, thus increasing trading activity. 
The trading rate at $t=0$ also depends on the length of the regulated period, with longer periods giving firms more time to offload and adjust their trading positions.
Trading increases with higher values of $\psi_i.$ Firms with higher $\psi_i$ can offset production reductions with permit trading, as the cost-effectiveness of purchasing permits may surpass the expense of further production cuts.

In Expression (\ref{eq: opt beta}), changes in the trading rate $d \hat \beta_i$ reflect adjustments in abatement, $d\hat \alpha_t^i,$ and net permit demand per unit of remaining time $\frac{1}{T-t}(\sigma_i dW^i_t - dM^i_t).$
Here, the term $(\sigma_i dW^i_t - dM^i_t)$ captures the difference between changes in uncontrolled emissions and the change in permit allocation.
Similar to the interpretation of $\hat \beta_0,$ changes in the trading rate also depend on the penalty values, the number of periods remaining, and fluctuations in the permit price.

\subsection{Market equilibrium and equilibrium price} \label{sec:equilibrium}

For a given permit allocation program $\bA,$ a market equilibrium is defined by a set of control processes $(\hat{\alpha}^i,\hat{q}^i,\hat{\beta}^i)_{i=1}^N$ and a price process $\hat P,$ such that the following condition is met:
\begin{equation}\label{eq: Ji}  J^i(\hat\alpha^i,\hat q^i,\hat\beta^i) = \inf_{(\alpha^i,q^i,\beta^i)}J^i(\alpha^i,q^i,\beta^i) \ \ \ \ \ \  i=1,\dots,N,
\end{equation}
where the equilibrium price $\hat P$ is the price at which the sum of permit trading across all firms, after considering their production adjustments, abatement strategies, and the allocation of permits, balances to zero, ie. $\sum_{i=1}^N \hat{\beta}^i_t=0 \ \ \forall t \in [0,T].$

\begin{Theorem}\label{Theo:equilibrium} \textbf{Equilibrium price} Given the cumulative allocation program $\bA$, the equilibrium price $\hat P$ is the unique solution to the Cauchy problem:
\begin{equation}
\label{eq: dynamics P equilibrium}
d\hat{P}_t = f(t)(d\bar{W}_t-d\bar{M}_t), \quad \hat{P}_0= f(0) \left[\big( \bar{H}+\bar{\mu}_b \big)T -\bar{M}_0\right],
\end{equation}
in which 
$$ f(t) := \frac{2\lambda}{1+2\lambda (\bar\eta + \bar\psi) (T-t)}, $$
where $\bar M$ and $\bar W$ represent the average (expected) allocation of permits and the average shock, respectively:
$$ \bar{M}_t :=\frac{1}{N}\sum_{i=1}^N M^i_t = \frac{1}{N}\sum_{i=1}^N \E_t[A^i_T],, \ \ \ \ \ \ \ \ \bar{W}_t :=\frac{1}{N}\sum_{i=1}^N \sigma_i W^i_t .$$

\end{Theorem}

\textit{Proof:}  See the Appendix  \ref{proof:Theo equilibrium}.

\vspace{0.15cm}
In the Cauchy problem (\ref{eq: dynamics P equilibrium}) the term $(d\bar{W}- d\bar{M})$ represents the difference between the change in uncontrolled emissions and the change in the average expected cumulative allocation of permits. 
This term determines the direction and dynamics of the equilibrium permit price.
When firms experience significant positive individual shocks, causing the average shock to exceed the change in the average allocation of permits $(d\bar{W} - d\bar{M} > 0)$, the permit price rises. 
This increase reflects the market's response to higher-than-expected emissions, driving up demand for permits.
Conversely, suppose there is a substantial policy change that significantly alters the expected future allocation of permits -- the difference $(d\bar{W}-d\bar{M})$ could become negative. 
In this case, the price change is negative, and the permit price decreases. 
This reflects a situation where an increase in the expected permit supply reduces the scarcity of permits, thereby lowering their price.

The permit price also depends on several other factors. The opportunity cost of non-compliance, $\lambda$. A higher penalty increases the urgency for firms to secure permits, which can drive up the permit price.
The price also depends on the average intensity of regulated goods,  $\bar{\psi} = \frac{1}{N}\sum_{i=1}^N \frac{\gamma_i^2}{\delta_i+2b_i}.$ 
This factor captures the average intensity of emissions, adjusted for the combined effects of cost and market conditions on production, reflecting how sensitive the equilibrium price is to changes in production and emissions intensity across the regulated firms. 
It also depends on the average curvature of the abatement function, $\bar \eta = \frac1N  \sum_{i=1}^N \eta_i.$ This represents the average sensitivity or reversibility of abatement across all firms. A higher $\bar \eta$ suggests that firms are either more responsive to changes in permit prices, or that they have more flexible abatement technologies, which can reduce the permit price.

The equilibrium price also depends on the distance to the end of the regulated phase. As this phase draws closer, potential changes in the permit price can have a more pronounced effect. This becomes particularly significant as firms and the market adjust their strategies in anticipation of final compliance requirements.\footnote{Historical evidence, such as the price fluctuations observed at the end of Phase I of the EU Emissions Trading System—largely due to the restriction on transferring allowances to Phase II—serves as empirical proof of this phenomenon.}

The starting level of the equilibrium price, $\hat{P}_0,$ varies depending on the extent to which BAU emissions over the entire period, $\bar{\mu}_{b} \cdot T$, exceed the average permit allocation, $\bar{M}_0.$ 
The average sensitivity of abatement costs, $\bar{H},$ further amplifies this price upward.
This initial price establishes the baseline for market evolution over the regulatory period, reflecting the initial balance (or imbalance) between expected emissions and the allocation program.

\hs
\noindent{\bf Example}: We now use an example to illustrate how the allocation program affects controlled emissions,  which are determined by firms’ optimal decisions to abate emissions, adjust production, or trade permits, and equilibrium prices. 
Later, we continue with this example and calculate the theoretical policy-driven impact on goods prices resulting from achieving net-zero emissions. 
Finally, we use the same example and evaluate the trade-off between emission reduction and price impact in the pursuit of net-zero targets.

\vspace{0.15cm}
\noindent First, recall that the dynamics of controlled emissions are described by relation~\eqref{eqn:emissions}, where $\hat q^i_t$ and $\hat\alpha^i_t$ are defined by~\eqref{eq: opt q}. 
By taking expectations and averaging over all firms, the per-firm expected controlled emissions at the terminal time $T$ can be expressed as:
\begin{align*}
\frac1N \E[E_T] = T (\bar \mu_b - \bar\psi P_0 - \bar \alpha_0)
\end{align*}
where $\bar \alpha_0 = \frac1N \sum_{i=1}^N \hat\alpha^i_0$ represents the average abatement effort across firms at $t=0$.
Since both abatement rate and production rate are martingales, we can re-write the per-period, per-firm average emissions at the terminal time $T$ as:
\begin{align*}
\frac1{N} \frac1{T} \E[E_T]  =  \bar\mu_b + \bar H - (\bar\eta + \bar\psi) P_0 
\end{align*}
which shows that controlled emissions primarily depend on BAU emissions $(\bar\mu_b)$ and the average sensitivity of abatement costs $(\bar H)$. 
Additionally, they depend on the sensitivity of permit prices to the combined effect of abatement costs and emissions intensity $(\bar\phi := \bar\eta + \bar\psi)$. 
This reflects the average flexibility of abatement measures across all firms, as well as the average intensity of emissions, adjusted for production costs and market conditions.

By substituting the expression for the initial permit price $P_0$, we obtain: 
\begin{align*}
\frac1{NT} \E[E_T] = \frac{(\bar\mu_b + \bar H)}{1 + 2\lambda(\bar\eta + \bar\psi) T} +
 \frac{2\lambda (\bar\eta + \bar\psi) T}{1 + 2\lambda(\bar\eta + \bar\psi) T} \bar M_0.
\end{align*}
The per-period, per-firm average controlled emissions essentially depend on three factors: BAU emissions, the average sensitivity to abatement costs, and the expected allocation of permits $(M_0)$. 
As the second coefficient describing the sensitivity of controlled emissions to the allocation program approaches 1 --with either an increase in penalty levels or a longer regulatory period-- controlled emissions become more directly tied to the allocation program.
This is the intended outcome in cap-and-trade systems, where the stringency of the regulatory policy is closely tied to the allocation program, ensuring that cap act as the primary driver of controlled emissions and play a central role in achieving emissions reductions.

\vspace{0.15cm}
To illustrate the sensitivity of controlled emissions to the expected allocation of permits and the corresponding cap, we consider Europe's net-zero emissions goal for 2050. 
Achieving this target requires reducing the baseline emission rate of production activities covered by the cap-and-trade progrma to zero.
Presently, this baseline stands at approximately $\bar\mu_b \approx 1.5\,10^9$~tonCO$_2$ per year, as per data reported by the European Environment Agency.
Using estimates from \cite{Metcalf23}, we evaluate the parameter $\bar\phi := \bar\eta + \bar\psi.$
Their analysis suggests that a \$40/ton carbon tax leads to a cumulative 5\% reduction in emissions. 
Applying this reduction to the total emissions covered under the European cap-and-trade program gives a per-dollar impact of $\bar\phi = 0.05 \times 1.5\,10^9 \div 40$ $=1.875\,10^{6}$~tonCO$_2$/(\euro$\cdot$y).
Next, we estimate the penalty parameter $\lambda$ by setting a maximum allowable discrepancy of 0.01 Gt in meeting the emission reduction goal. 
Based on this condition, and following \cite{Aid-Bia}, we derive $\lambda$ $=1.25\,10^{-6}$~\euro/ton$^2.$
With these parameters, we calculate the sensitivity of controlled emissions to the allocation program $M_0$:
\begin{align*}
\frac{2\lambda (\bar\eta + \bar\psi) T}{1 + 2\lambda(\bar\eta + \bar\psi) T} = \frac{1}{1 + \frac{1}{4.7 T}}.
\end{align*}
This expression shows that as the regulatory period extends beyond $T > 2,$ the allocation program becomes increasingly influential, reinforcing its role as the primary driver of controlled emissions in the cap-and-trade system.

We now calculate the sensitivity of the initial permit price $P_0$ to the same model parameters:
\begin{align*}
\frac{2\lambda}{1 + 2\lambda(\bar\eta + \bar\psi) T} = \frac{2.5\,10^{-6}}{1 + 4.7 T}.
\end{align*}
Setting the regulatory horizon at 10 years, this corresponds to a permit price sensitivity of approximately $5\, 10^{-8}$ (\euro/ton)/ton. This means that a change in the allocation program by 100 million tons would shift the carbon price by $5$~\euro/ton.
Shortening the regulatory period to 5 years doubles this effect, causing a price movement of $10$~\euro/ton. 
Similarly, if $T$ represents the prevailing investment horizon of market participants, as described in \cite{quemin-baudry}, a more myopic market outlook --reducing the horizon from 10 to 5 years-- would likewise double the impact, leading to a $10$~\euro/ton price shift.

\section{Inflationary impact of carbon pricing regulation}\label{ssec:inflation-impact}

Recent debates have focused on the potential inflationary effects of carbon pricing policies on the prices of emission-intensive goods and services.
The literature offers mixed evidence on the inflationary effects of carbon pricing. Some empirical studies suggest that the impact on inflation in Europe and Canada has been relatively weak (\cite{Metcalf2019}; \cite{Metcalf23}; \cite{Konradt23}), while other research, particularly on the European Union Emissions Trading System (EU ETS), indicates larger inflationary effects (\cite{Kanzig23}). 
These differences can be attributed to variations in sectoral composition and historically low carbon prices.\footnote{The EU ETS targets high-emitting industries, which are more likely to pass on emission costs to consumers (\cite{fabra-reguant}). In contrast, national carbon taxes examined by \cite{Metcalf2019}, \cite{Metcalf23}, and \cite{Konradt23} often focus on sectors like transportation, which cover fewer emissions. In these sectors, higher costs may be absorbed by firms or avoided by households through fuel or transport substitution, possibly explaining the differing inflationary impacts across studies. Additionally, historical carbon taxes and permit prices in cap-and-trade systems have generally been low and less ambitious compared to upcoming goals like the EU’s ``Fit-for-55" targets for 2030. Simulation studies indicate that a significant increase in permit prices  -- potentially reaching 150 euros per ton of carbon by 2030 -- could raise inflation by 0.2 to 0.4 percentage points annually until 2030. Consequently, past data on carbon prices may not reliably predict future inflationary effects.}

\vspace{0.15cm}
We use our model to establish a simple theoretical foundation for assessing the inflationary impact of carbon pricing, particularly in scenarios involving more ambitious emission reduction targets and the resulting higher permit prices.
Specifically, we examine the impact of a cap-and-trade system on the Consumer Price Index (CPI).\footnote{The model focuses on real prices rather than nominal prices. In the context of carbon pricing, the primary concern is the adjustment of relative prices due to emission costs, which influence real production and consumption decisions.}
By comparing the policy-driven changes in CPI with the CPI under a BAU scenario, we quantify the price impact of carbon pricing and calculate the resulting theoretical inflationary pressure.
In this model, the CPI comprises $N$ distinct regulated goods, each produced by a different sector $i$, such as energy, steel, cement, and others. 
Thus, each $i$ represents a regulated sector producing a single good.\footnote{The scenario where only a portion of the CPI is affected by carbon pricing, or where multiple regulated firms produce each good, can easily be extended from the case described here.}
For a given permit allocation $\mathbf A$ and equilibrium permit price $\hat P$, we define the CPI as a weighted sum of the prices of the $N$ goods under BAU, $\hat \pi_t,$ and policy-driven, ${\pi}_b$:
\begin{align} \label{def: consumers basket}
\hat \pi_t := \sum_{i=1}^N w_i S_i(\hat q^i_t), \quad {\pi}_b := \sum_{i=1}^N w_i S_i(\tilde q_i), \quad  w_i \in (0, 1), \quad \sum_{i=1}^N w_i =1,
\end{align}
where $\tilde q_i$ represents the optimal production without regulation, while $\hat q^i_t$ denotes the optimal production under the cap-and-trade system. The latter is obtained by substituting the equilibrium price from \eqref{eq: dynamics P equilibrium} into the production adjustment Equation \eqref{eq: opt q}.

\vspace{0.25cm}
Typically, the weights $w_i $ reflect the relative importance of each good in the CPI. Thus these weights are determined by the proportion of each good in the overall consumption basket, ensuring that goods with higher consumption shares have a greater influence on the CPI. 
Beyond inflation, the impact of carbon pricing regulation can be assessed by considering alternative weighting criteria that reflect household income differences.\footnote{CPI does not account for differences in household income levels, but accounting for these differences is important. According to a recent report \cite{IMF2022} report, increases in energy prices disproportionately affect low-income relative to high-income households. For example, in the UK, energy price increases contribute to a 6\% rise in the cost of living for the wealthiest 20\% of households, but a 15\% rise for the poorest 20\%. This disparity is even more pronounced in Eastern European countries, where energy inflation exposure for the lowest-income households can reach up to 25\%.} 
For example, by adjusting the weights $w_i$ in the CPI to reflect the relative spending of different income brackets, we can use this simple theoretical framework to analyze the policy-driven impact on an income-stratified CPI or cost-of-living indices.
These indices capture the effect of carbon pricing by incorporating both price changes and shifts in consumption patterns based on income, without requiring any changes to the underlying formulas presented in this section.
Adjustment of the weights would likely capture the broader social impact of carbon pricing and provide a more nuanced measure of its distributional effects.

\vspace{0.25cm}
The policy-driven CPI adjustment is:
\begin{align} \label{eq:polindCPI} 
{\hat \pi_T - \pi_b} & = \sum_{i=1}^N w_i (S_i(\hat{q}^i_T)-S_i(\tilde q_i)) 
=  \underbrace{\sum_{i=1}^N \frac{b_i}{\delta_i+2b_i} \gamma_i w_i}_{\bar \omega} \hat P_T
\end{align}

\vspace{0.15cm}
The \textit{expected average inflation rate} over the period $T,$ representing the expected per-period inflation rate $I,$ is given by:
\begin{align}  \label{eq:inflation} 
I := \frac{1}{T} \frac{\E[\hat \pi_T] - \pi_b}{ \pi_b} = \frac{1}{T} \frac{\bar \omega}{\pi_b} \hat P_0
\end{align}

\vspace{0.15cm}
\noindent \textbf{Example continued:}
We continue the previous example to calculate the theoretical policy-driven impact on goods prices associated with achieving net-zero emissions by a specified future year $T.$ 
This calculation establishes a baseline for understanding the potential price-level implications of ambitious carbon reduction targets.
This objective necessitates reducing the average emissions drift to zero:
\begin{align} \nonumber
\E[\bar \mu_T] = 0, \quad 
\bar \mu_T:= \frac{1}{N} \sum_{i=1}^{N} \mu^i_T, \quad
\mu^i_T:= \gamma_i \hat q^i_T - \hat \alpha^i_T.
\end{align}
This condition implies that the average emissions drift across all firms is neutralized by the end of the regulatory period.
To meet this target, if some firms exceed their emission limits, others will need to drastically reduce their emissions, potentially capturing more emissions than they produce and acting as net sequesters, thereby ensuring that overall emissions balance out to achieve the net-zero goal.
By substituting the optimal production adjustments and emissions abatement levels required for each firm to meet the net-zero target, $(\hat q^i_T, \hat \alpha^i_T)$, we can calculate the initial equilibrium price and the average inflation rate induced per year by this ambitious goal:
\begin{align*} 
\E[\hat P_T] = \hat P_0 = \frac{\bar \mu_b + \bar{H}}{\bar\eta + \bar\psi}, \quad
I = \frac{1}{\pi_b } \frac1T \frac{\bar\mu_b + \bar H}{\bar\eta + \bar\psi} \bar \omega.
\end{align*}

We first calculate the initial permit price for the scenario where the net-zero emissions goal is applied to the European cap-and-trade program by 2050. Achieving this target requires reducing the baseline emission rate, $\bar\mu_b \approx 1.5\,10^9$~tonCO$_2$ per year, to zero.
Using the parameter values from the previous example, we set $\bar\phi =1.875\,10^{6}$~tonCO$_2$/(\euro $\cdot$ y).
Additionally, we select $\bar H \approx \bar \eta h = 6 \,10^6 \times 25 = 1.5 \,10^8$, based on illustrative estimates from \cite{Gollier2024} and the MIT Emissions Prediction and Policy Analysis model \cite{Morris2012}.
Using these parameters, we calculate an initial carbon price of $P_0 = 880$~\euro/ton. 
Notably, this estimate aligns closely with carbon price projections utilized by public authorities, such as those outlined in France Stratégie's 2019 report. The report estimates a carbon value starting at 250\euro/ton in 2025 and rising to 775~\euro/ton by 2050, intended to guide decision-making processes related to evaluating avoided emissions.\footnote{See Alain Quinet, {\em The Value for Climate Action}, France Stratégie, 2019, Chapter 4, Section 3.2, Figure 42, which presents a carbon value starting at 250\euro/ton in 2025 and ending at 775~\euro/ton in 2050.}
Combining this price estimate with findings in \cite{Konradt24} that a carbon price of 40~\euro/ton leads to a total excess headline inflation of 3\% over 10 years ---equivalent to an average annual inflationary impact of 0.3\% per year--- suggests that achieving net-zero emissions by 2050 with a carbon price of 880~\euro/ton could result in an annual excess inflation rate of approximately 6.6\%. This estimation leads to consider $\bar\omega = 0.3/40 = 7.5\,10^{-3}$ of percentage of inflation per \euro/tonCO$_2$ of carbon price.

\section{Balancing environmental and inflation concerns: navigating policy trade-offs}
\label{sec:regulation}

In a period of high inflation, balancing inflation concerns with environmental targets can become essential for making carbon pricing policies more acceptable.\footnote{Many commentators have raised concerns that more ambitious climate policies, which could significantly increase prices, risk triggering Yellow Jackets-style backlash and broader political resistance. See, for example, discussions in Politico (December 16, 2024), Bloomberg (June 10, 2024), and The Financial Times (March 28, 2024).}
Imagine if an authority could decide how much inflationary pressure they are willing to tolerate to achieve ambitious carbon reduction goals. 
This approach would involve a careful trade-off: pressing too hard on emissions reductions could drive substantial inflation, while excessive caution on inflation could hinder progress toward pollution control targets. 
While in practice, we live in a world with an independent central bank focusing on inflation and a separate regulator for carbon pricing, here we consider a single carbon pricing regulator who accounts for the political acceptability of the policy by considering its inflationary impacts alongside emission reduction goals.

We explicitly model this policy trade-off in the regulator’s objective function, where the goal of minimizing firms' compliance costs is balanced by managing deviations from both a pre-set emission reduction target and an acceptable inflation target.
This approach can be represented as a weighted function that penalizes deviations from these dual targets.
The new regulator's objective is:
\begin{align}
\label{eq:regprob}
&\inf_{\bA} J^{\rm R}(\bA) := \sum_{i=1}^N  J^i(\hat \alpha^i,\hat q^i,  \hat \beta^i)  +   \E\left[\ell( \bar \mu_T - \theta)\right]  +   \E\left[\varphi\left ( \frac{1}{T}\frac{\hat \pi_T - \pi_b}{\pi_b} -  \nu\right)\right],
\end{align}
where $\ell$ is a real function representing the convex penalization of the deviation between the policy's projected average emissions drift and the target average emission rate $\theta.$ This function captures the regulator's goal of ensuring that the mean emissions reduction efforts are on track to meet the environmental targets. A larger deviation from $\theta$ results in a higher penalty, reflecting the importance of staying close to the desired emission path.
$\varphi$ is a real function representing the convex penalization of the policy-induced inflation compared to an acceptable inflation rate $\nu$ (e.g. $\nu = 2$\%/year). This function embodies the regulator's concern with maintaining reasonable inflation and policy acceptability while implementing carbon pricing policies. If the inflation caused by the policy exceeds the acceptable threshold $\nu$, the penalty increases, signaling the need to mitigate the inflationary impact.

\vspace{0.15cm}
We solve the regulator’s problem in a general setting, where the penalization functions $\ell$ and $\varphi$ can take any convex form. 
After establishing the general solution, we illustrate a specific case of quadratic penalization through a numerical example.
To simplify the notation for all optimal quantities that solve the regulator's problem, we adopt a streamlined approach: we drop the hat from the variables and denote the optimal solutions with a superscript $*.$ 
For example, instead of writing $\hat P^*$ for the optimal equilibrium price, we will write $P^*$. This notation will be applied consistently across all relevant variables, making the expressions easier to follow throughout the analysis.
 
\begin{Theorem}\label{Theo:ell-phi}
\textbf{Solution to the regulator’s problem \eqref{eq:regprob}.} The allocation process $\bA^* =(A^{1,*}, \ldots, $ $ A^{N,*})$  is optimal if and only if it satisfies
\begin{align} \label{eq: opt allocation}
d\bar M^*_t&=d \bar W_t  \quad \quad \bar M^*_0 = T\bar{H} -
      P^* \left(\frac{1}{2\lambda} +  \bar{\phi} T\right )
\end{align} 
where $\bar M^\ast_t = \E_t[\bar Y^\ast_T]$,    $\bar{\phi}$ is the average of the emission drift resulting from BAU  production, and $P^* = P_t^* $ is the unique, constant equilibrium price induced by such allocation. The price is found as the unique minimum of the real  function $s$: 
\begin{align*} 
s(x) & = x^2   \frac{NT}{2} \left( \bar{\phi} +\frac{1}{2\lambda T}  \right)    + \ell\left( \bar{\mu}_b-\bar{\phi}\, x + \bar H-\theta \right)
+     \varphi \left(  \frac{\bar \omega x}{T \pi_b}-\nu\right ). 
\end{align*}

\noindent The optimal production and optimal abatement rates are:
\begin{align*}
 {\alpha}^{i,\ast}_t = \eta_i(P^* -h_i)  \ \quad  {q}_t^{i,\ast}  = \tilde q_i   -\frac{\psi_i}{\gamma_i}P^\ast \  \quad \ \forall t \in [0,T].
\end{align*}

The optimal average emission rate at the terminal time, the policy-induced CPI, and the corresponding per-period rate of inflation are:
\begin{align*}
{\bar\mu}^\ast_T = \bar{\mu}_b +\bar H - \bar{\phi}P^\ast,\quad \quad
\hat \pi_T = \pi_b + \bar \omega P^*, \quad \quad 
\hat \imath^\ast_T =  \frac{\bar{\omega}}{T\pi_b} P^\ast.
\end{align*}

\end{Theorem}

\textit{Proof:}  See the Appendix  \ref{app: Section 5}.

\hs
The expressions for ${\bar\mu}^\ast_T$ and  $\hat \pi_T$ clearly illustrate the trade-off between emission reduction and price impact. 
As the permit price $P^*$ rises, its downward pressure on emissions trend ${\bar\mu}^\ast_T$ increases, achieving greater emission reductions. 
However, this also drives the consumer price index $\hat \pi_T$ higher, ultimately leading to greater policy-driven inflationary pressure from stringent emission controls.
This trade-off depends on the regulator's priorities. If the regulator places more weight on emission reductions over inflation control, the optimal policy will focus on further reducing emissions, as observable in the expression $s(x).$ This results in a higher permit price, pushing ${\bar\mu}^\ast_T$ closer to the target average emission rate but at the expense of raising $\hat \pi_T$ and $\hat I^\ast_T,$ which reflect the greater price and inflationary impact.

\noindent \textbf{Example continued:}
We continue the previous example to illustrate the trade-off between emissions reduction and policy-driven price impacts, focusing on achieving net-zero emissions in Europe by 2050 with a regulatory horizon of $T=25$ years. 
By adjusting the penalization parameters for deviations from emission reduction and inflation targets, we check empirically how the regulator balances these competing objectives. This analysis highlights how the relative weight assigned to inflation and emissions targets influences equilibrium outcomes, including the permit price and overall policy costs. 
First, the average emission rate at the terminal time is given by:
\begin{align} 
\E[\bar \mu_T] &= \bar \mu_b + \bar H -  \bar{\phi}  \hat P_0,  
\end{align}
from Expression \eqref{eq:inflation}, expected average inflation rate over the period $(0,T)$ is
\begin{align*} 
    I  &= \frac{\E[\hat \pi_T] - \pi_b}{T \pi_b} =   \bar \omega \frac{\hat P_0}{T \pi_b},  
\end{align*}
the optimal permit price is
\begin{align} \label{eq:Popt}
{P^*} = {\frac{ \bar{\phi}(\bar\mu_b + \bar H -\theta)y_\mu + y_\pi \frac{\bar \omega \nu}{T \pi_b}}{  NT( \frac12\bar{\phi}+\frac{1}{4\lambda T}) + \bar{\phi}^2 y_\mu +   \frac{\bar \omega^2}{T^2 \pi_b^2}  y_\pi}}, 
\end{align}
where $y_\pi$ is the penalty parameter for deviations from the inflation target, and $y_\mu$ is is the penalty parameter for deviations from the emission reduction target.
The regulator's objective cost function depends on the permit price $x$ 
\begin{align} \label{eq:cost}
s(x) & =  s(0) 
+  N T \left (\frac12{\bar{\phi}} + \frac{1}{4\lambda T}\right)  x^2 + y_\mu \left( \bar{\mu}_b- \bar{\phi}\, x +\bar H -\theta  \right)^2+ y_\pi \left(\frac{\bar \omega x}{T\pi_b}-\nu \right)^2. 
\end{align}
We summarize all parameters estimated in the preceding sections of the example in Table~\ref{tab:params}. 
\begin{table}[h]
\centering
\begin{tabular}{c c c c c c c c } 
T & $\bar\mu_b$ & $\bar H$ & $\bar{\phi}$ & $\bar\omega$ & $y_\pi$ & $y_\mu$ & $\lambda$      \\ \hline
y &Gt/y  & Gt/y & Mt/(\euro$\cdot$y) &  \%/(\euro/t)   & G\euro/(\%/y)$^2$  &  \euro/(tCO$_2$$\cdot$y)$^2$     & \euro/ton$^2$   \\ \hline
25 & 1.5 &  0.15 & 1.875 &  0.0075 & 750 & $10^{-3}$  &  $1.25\,10^{-6}$   \\ \hline
\end{tabular}
\caption{Parameter values with their corresponding units.}
\label{tab:params}
\end{table}

To calibrate the penalty parameter for deviations from the inflation target ($y_\pi$) we look for reasonable value for $y_\pi$ from economic research that study GDP-inflation elasticity.
Drawing on a recent survey of OECD countries (\cite{Kirsanli24}),  use an average estimate indicating that an annual excess inflation rate of 1\% results in a GDP contraction of 0.1\%.\footnote{This magnitude aligns with typical values generated by macroeconomic models used by central banks to estimate the impact of inflation shocks on the economy. For example, see \cite{lemoine2018fr}, non-technical summary, p.~iii.} 
Based on this relationship, we assume that the marginal cost of excess inflation, represented in our model by $2 y_\pi (i - \nu),$ equates to $10^{-3}$ of the EU's GDP. With the EU's GDP estimated at approximately \euro$1.5 \times 10^{13}$, we calculate $y_\pi$ as:
$y_\pi = 0.5 \times 10^{-3} \times 1.5\,10^{13}$ $=750$~G\euro/(\%/y)$^2$. 
Now, regarding the penalty parameter $y_\mu$ of the emission rate target, we observe that the penalisation is concerned about the square of the emission rate. Thus, even a deviation of $100$ millions tons of emissions (ie close to 10\% of the baseline emission rate of $10^9$ tons per year), would yield a cost of $y_\mu 10^{16}$. A value of $y_\mu =10^{-3}$ would basically resolves in a social cost of $10^{13}$ which is of the order of magnitude of the European GDP. For an emission rate deviation of $10$ millions tons (1\% of the baseline), the cost would be $10^{11}$, ie 1\% of the European GDP. Thus, considering $y_\mu$ to be of the order of $10^{-3}$ is the largest possible value we can take for a reasonable penalisation of the emission target deviation. Besides, because of the uncertainty surrounding the values of the penalty parameters $y_\pi$ and $y_\mu$, we develop a sensitive analysis around these reference values. For the reference value of the parameters given in Table~\ref{tab:params}, the optimal price of emissions $P^\ast$ given by relation~\eqref{eq:Popt} is $874$~\euro/tonCO2, ie very close to the price needed to achieve net-zero in 25 years.

\begin{figure}[hbt!]
\begin{center}
\begin{tabular}{c c}
(a) & (b) \\
\includegraphics[width=0.45\textwidth]{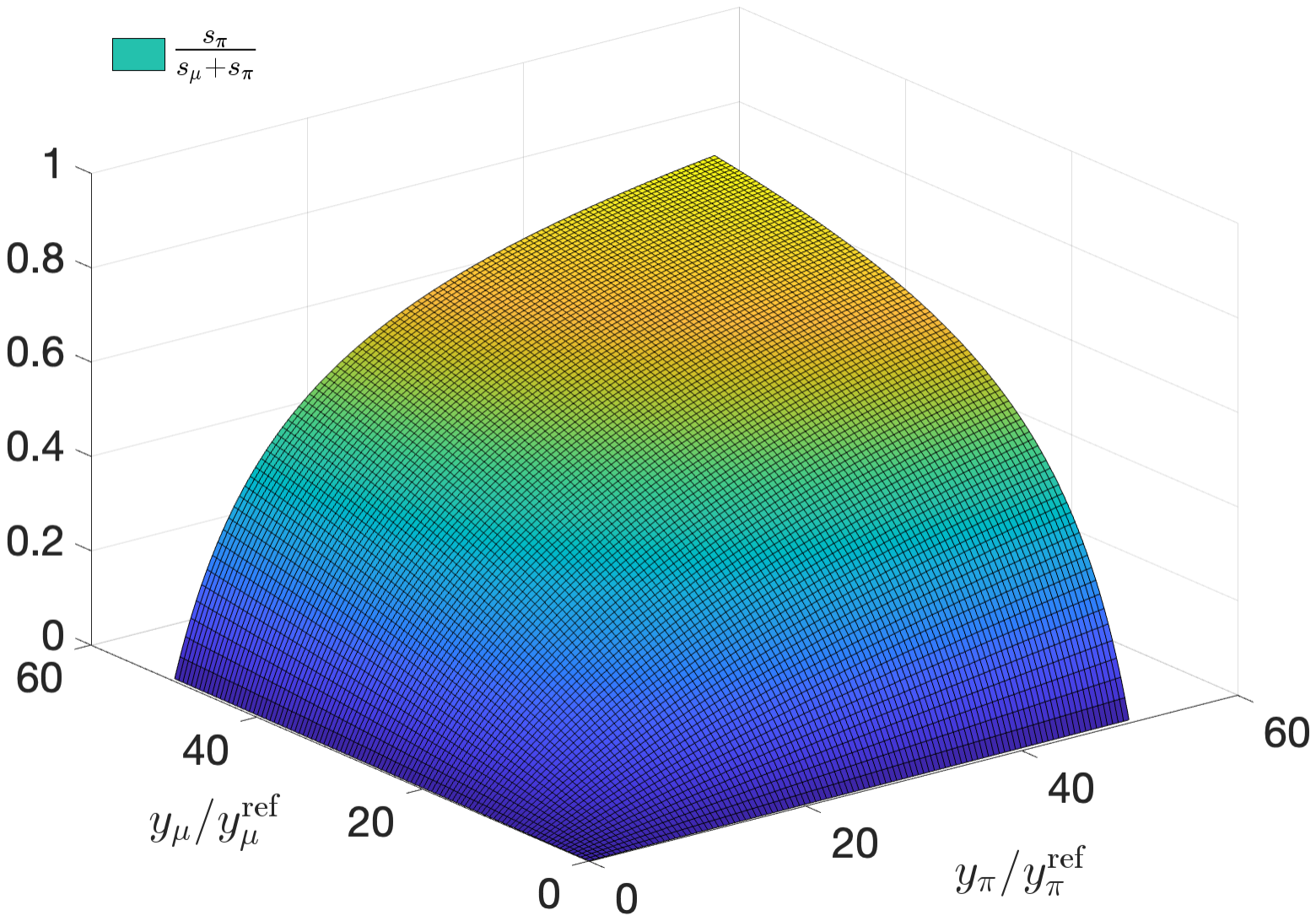}&
\includegraphics[width=0.45\textwidth]{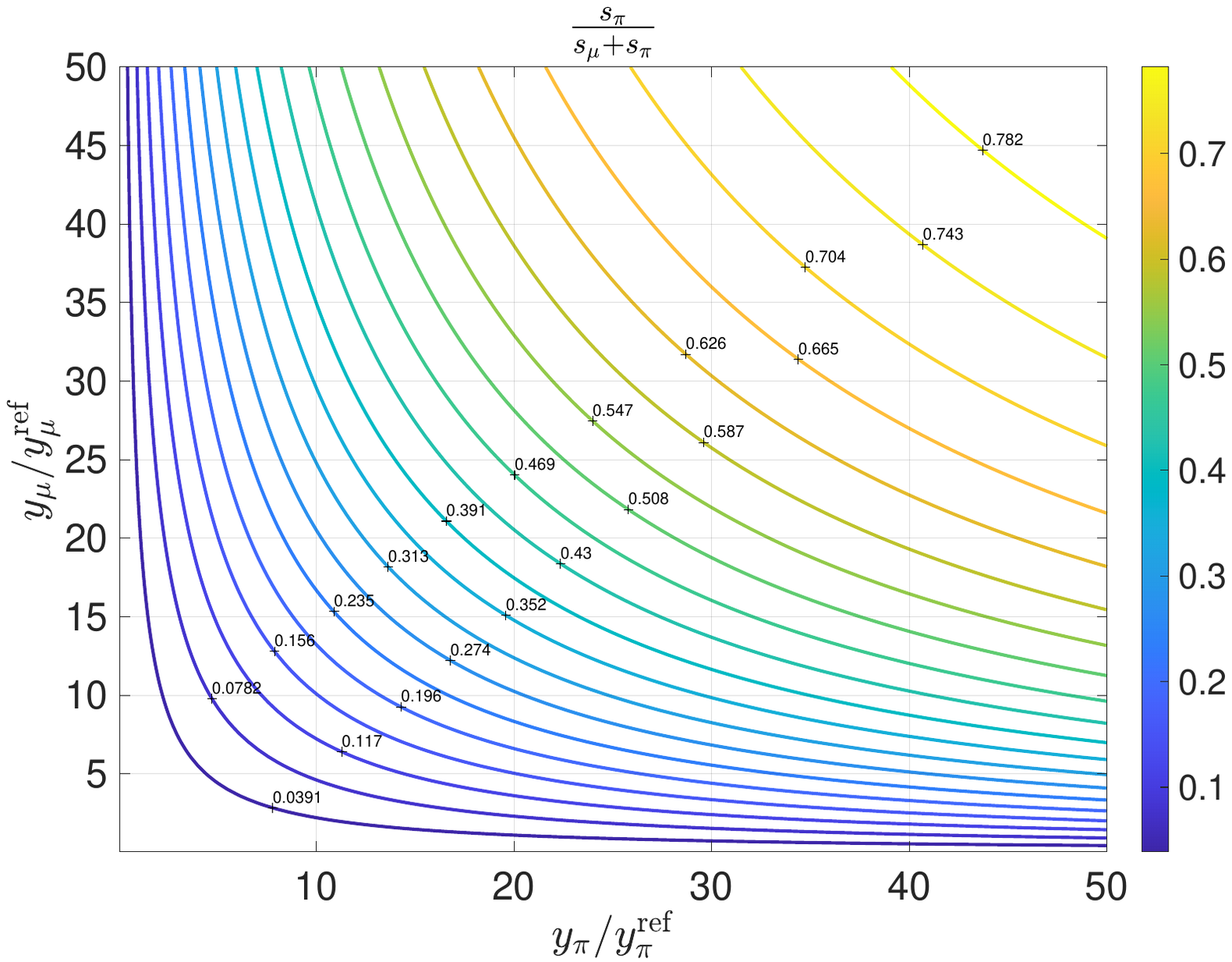}\\
(c) & (d) \\
\includegraphics[width=0.45\textwidth]{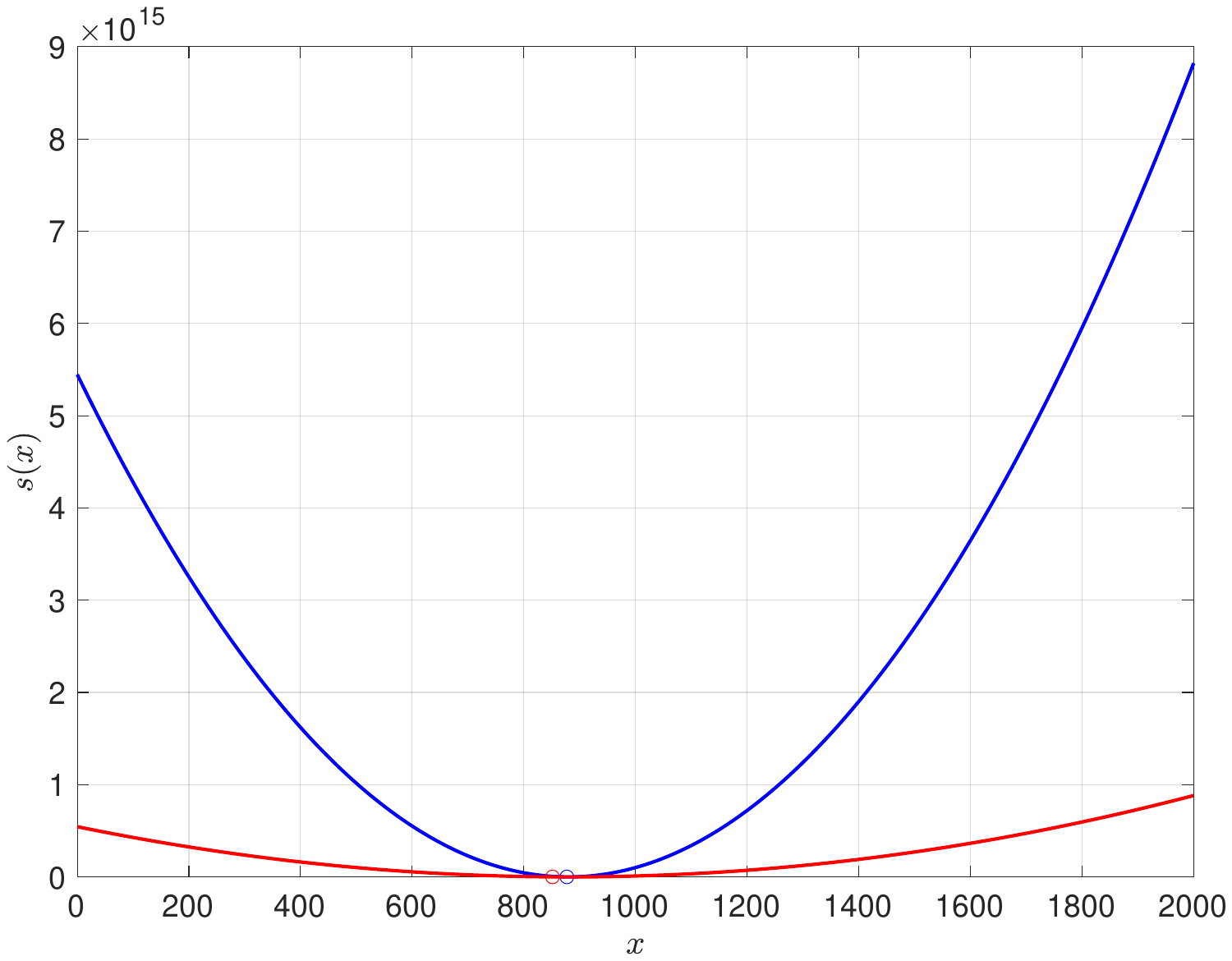}&
\includegraphics[width=0.45\textwidth]{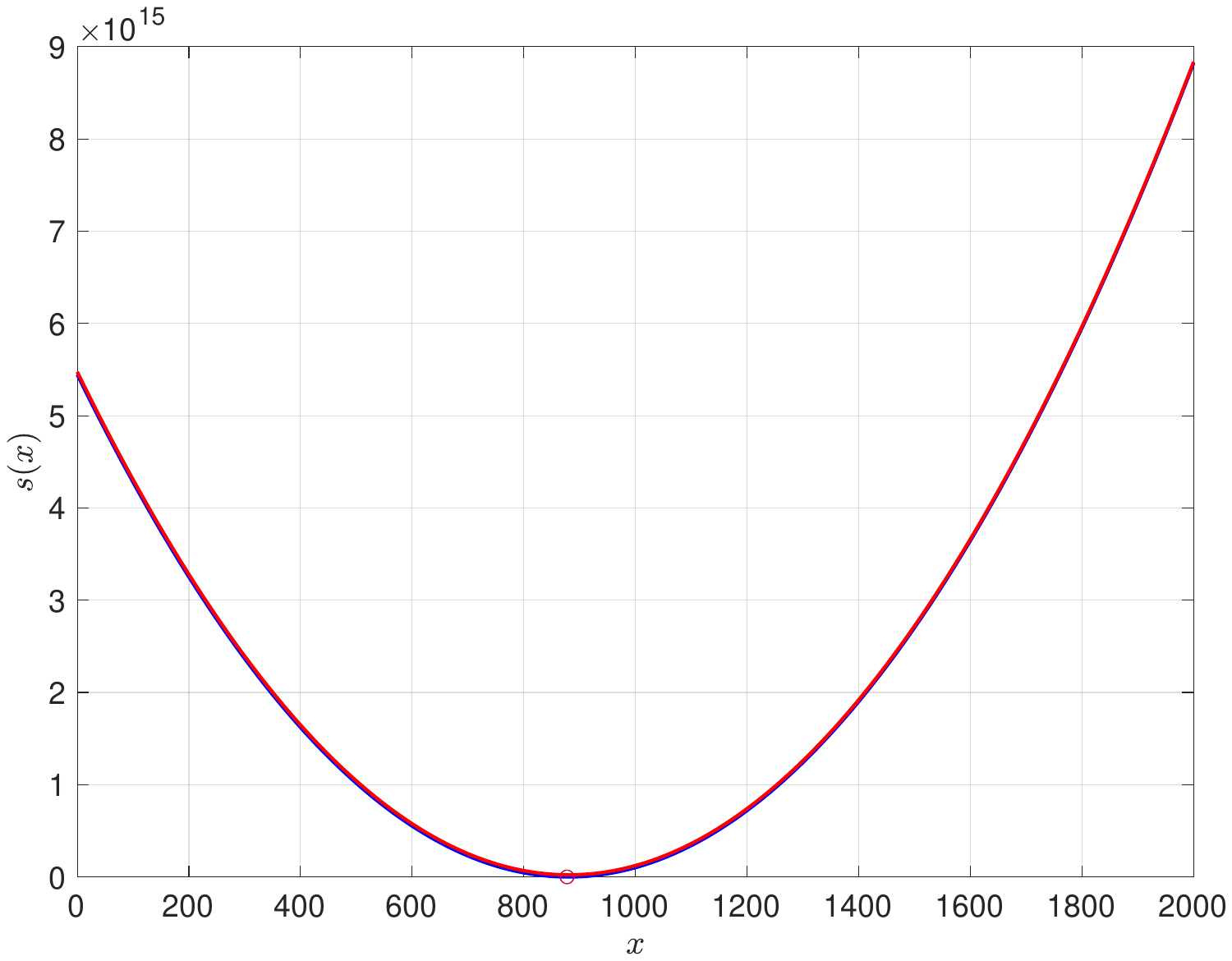}
\end{tabular}
\vspace{-3mm}
\caption{{\small Ratio $s_\pi(P^\ast)/\big(s_\mu(P^\ast) + s_\pi(P^\ast)\big)$ as a function of the penalties $y_\mu$ and $y_\pi$ as a surface plot (a) and as level lines (b).  The social cost function $s(x)$ as a function of the emissions price $x$ for two different values of $y_\mu$ (c) and for two different values of $y_\pi$ (d).}}
\label{fig:costmupi}
\end{center}
\end{figure}

We graphically illustrate the trade-off between inflation concerns and environmental targets by examining how the regulator’s objective cost function varies with changes in the penalization parameters.
For example, increasing $y_\pi$ reflects a scenario where the regulator is less tolerant of deviations from the acceptable inflation target.
For ease of notation, we denote the third and fourth terms of the regulator’s objective cost function as $s_\mu(x) := y_\mu \left( \bar{\mu}_b -  \bar{\phi}\, x +\bar H -\theta  \right)^2$ (costs associated with emissions deviations) and $s_\pi(x) := y_\pi \left(\frac{\bar \omega x}{T\pi}-\nu \right)^2$ (costs associated with inflation deviations).
Figure~\ref{fig:costmupi}~(a) and (b) plot the ratio $s_\pi(x) / (s_\mu(x)+s_\pi(x)),$ illustrating the relative weight of inflation-related costs compared to the total inflation- and emissions-related costs. Figure~(a) represents the ratio as a surface (a function going from $\mathbb R^2$ to $\mathbb R$) while Figure~(b) represents the level lines of the same function.
We consider variations in the penalization parameters over two orders of magnitude, achieved by multiplying and dividing the reference values $(y_{\cdot}^{\rm ref})$ in Table~\ref{tab:params} by 50, to provide a view of the sensitivity of the regulator’s objective to these parameters.
Considering extreme cases, when the inflation penalization parameter $y_\pi$ is 50 times larger and the emission target penalization parameter $y_\mu$ is 50 times smaller than their reference values, the relative relevance of inflation-related costs becomes 78\% of the total $s_\mu(x) + s_\pi(x)$. 
This indicates that even when inflation concerns are significantly amplified, and the emphasis on achieving environmental targets is relatively low, inflation does not entirely dominate the regulator's decision-making. 
This reflects that price stability cannot be fully prioritized over emissions control, as achieving some level of environmental compliance remains a fundamental aspect of the regulator's objectives.
Conversely, when $y_\mu$ is 50 times larger than its reference value while $y_\pi$ is 50 times smaller, the relative relevance of inflation-related costs drops to 4\%. 
In this case, the regulator’s focus shifts almost entirely toward achieving the emissions reduction target, with minimal weight given to deviations from acceptable inflation levels.
Interestingly, the relative weighting remains virtually unchanged even when the priority of adhering to the environmental target cap decreases. This outcome underscores the centrality of emission-related costs in the regulator's decision-making process, highlighting that, even when political pressures push the regulator to prioritise minimising inflation increases, emissions reduction goals should continue to carry significantly greater weight than concerns over inflation deviations.\footnote{Non-reported results considering alternative theoretical marginal damage costs of 250~\euro/t and 100~\euro/t for emissions exceeding the net-zero target deliver similar results. This finding emphasizes the robustness of emission-related priorities, even in scenarios where the marginal damage cost is lower.} 

\hs

Returning to the critical importance of controlling emissions, Figure~\ref{fig:costmupi}~(c) and (d) offer a comparison of how the regulator’s objective cost function, $s(x)$, changes in response to variations in the permit price $x$, under differing penalty parameters scenarios.
By plotting the cost function under reference values for the penalty parameters associated with emissions deviation and inflation deviation, the blue curves serve as a baseline. 
In Figure~\ref{fig:costmupi}~(c), the sensitivity of the cost function to changes in $y_\mu$ is evident. 
When the value of $y_\mu$ is reduced by an order of magnitude (divided by 10), there is a notable shift in the cost function, as depicted by the red curve. Once again, this substantial variation underscores the importance of $y_\mu$, the penalty parameter for emissions control, in shaping the regulator’s decisions. It reflects that even small adjustments to how deviations from environmental targets are penalized can significantly influence the overall cost dynamics.
Conversely, Figure~\ref{fig:costmupi}~(d) shows the response of the cost function to variations in $y_\pi$. Unlike the previous case, the regulator’s cost function is relatively insensitive to even drastic changes in $y_\pi$. For instance, increasing $y_\pi$ by five orders of magnitude ($10^5$) results in minimal change to the cost function, as illustrated by the near-identical curve. This stability suggests that the policy-driven inflation costs plays a secondary role compared to emissions-related costs in determining the structure of the cost function.

\section{Conclusion}

The primary aim of carbon pricing is to adjust the relative cost of emission-intensive goods, which can lead to an increase in the overall price level. 
In a period of elevated inflation and growing ambitions to mitigate climate change, potentially through broader carbon pricing coverage, these concerns have gained heightened political relevance.
Criticism of carbon pricing has featured prominently in political debates in countries such as Canada and Australia. In Canada, opposition campaigns under slogans like "axe the tax" reflect growing resistance to such measures. Similarly, in Australia, carbon pricing policies have faced significant public and political pushback. 
In Europe, where the future of carbon pricing remains a contentious issue, critics argue that it is punitive and overly costly. Commentators have raised concerns that more ambitious climate policies could provoke backlash akin to the Yellow Jackets protests or broader political resistance, particularly in countries like France and Germany.
These developments have prompted policymakers and academics to closely examine the effects of carbon pricing on inflation and consider how to reshape policy approaches by balancing carbon reduction goals with inflationary pressures to maintain public and political support. The empirical literature on this issue offers mixed evidence, with no clear consensus.
We address this debate by developing a theoretical framework to evaluate the trade-off between achieving ambitious emission reduction targets and managing inflationary pressures induced by carbon pricing. Our model reveals that, even when significant emphasis is placed on maintaining inflation at acceptable levels or when emissions reduction targets are given lower priority, the costs associated with emissions deviations far outweigh any savings from marginally lower inflation. Consequently, emission reduction goals should remain the primary focus for policymakers, underscoring their central role in designing carbon pricing policies.

We consider an economy governed by a cap-and-trade system, where firms comply with emission caps through production adjustments, abatement, and permit trading.
Firms' optimal compliance strategies are shaped by the stringency of the allocation program, which directly impacts the consumer price index and drives inflation.
Changes in production, driven also by firms' choices between abatement and trading, emerge as the primary channel through which carbon pricing influences inflation.
The equilibrium permit price ultimately reflects the regulated supply of permits set by the allocation program.
Using an illustrative example grounded in real data, we calculate the theoretical inflationary impact of achieving net-zero emissions in Europe by 2050. This baseline analysis highlights the price-level implications of ambitious carbon reduction targets, particularly in the absence of explicit considerations for inflationary effects.
We then explore the trade-off between emissions reduction and inflation by introducing a regulator’s objective function that balances environmental and inflationary concerns. 
By varying the weight assigned to deviations from both emission reduction and inflation targets, we illustrate how regulators can effectively manage these competing priorities in practice. 
Our empirical results reveal that, under reasonable model parameterization, even when significant weight is placed on maintaining inflation near acceptable levels or when the emphasis on achieving emissions reduction targets is substantially relaxed, emission-related costs outweigh any savings from marginally lower inflation.
This result has clear policy implications: even with a strong focus on controlling inflation, emission reduction goals must remain paramount, as minor inflationary benefits cannot justify compromising on emission reduction objectives.

\appendix
\section{Appendix}

\subsection{Proof of Theorem \ref{Theo:single-firm-opt}} \label{proof:Theo single-fim-opt}

The firm's program  \ref{eq:objfirm}  is an unconstrained convex optimization,  with smooth objective function.  So, we simply find the FOC by annihilating the Frechet derivatives wrt the controls.  As in  \cite{Aid-Bia}  one gets:
\begin{equation*}
\label{eq:diffJi2}
 \begin{cases}
\frac{1}{\eta_i}\alpha^i_t + h_i + 2\lambda \E_t[X^i_T] =0 \\
-a_i + \kappa_i - \delta_i\tilde q_i + (\delta^i + 2 b^i)q^i_t -2\gamma_i \lambda \E_t[X^i_T]=0\\
P_t  +2\lambda \E_t[X^i_T] =0
\end{cases}
\end{equation*}
From the third equation above, it is apparent that the problem has a solution if and only if  $P$ is a martingale. We assume then that this is the case.       The FOC system gives: 
$$ P_t = -2\lambda \E_t[\hat X^i_T], \ \ \
\hat{\alpha}^i_t = \eta_i(P_t-h_i)  , \text{ and }
\hat{q}^i_t =   \frac{1}{\delta_i + 2 b_i} (-\gamma_i P_t + a_i
-\kappa_i+\delta_i \tilde q_i) = \tilde q_i  - \frac{\gamma_i}{\delta_i + 2 b_i}P_t.
$$
Since the problem lacks strict convexity in $\beta^i$, the above offers an equation only for the cumulative trade variable $B^i_T = \int_0^T\beta^i_t dt$, a component of $X^i_T$, for which we prove uniqueness. To this end, call $B^i_t =\E_t[\hat B_T]$,and focus on   $\E_t[\hat X^i_T]$. 
 \footnote{Note that $B^i_t =E_t[\int_0^T\beta^i_t dt ]\neq \int_0^t\beta _t dt$ in general, i.e the conditional expectation of the total trade is not the cumulated trade up to $t$. }
Then,
\begin{align*}
\E_t[\hat X^i_T] &=  -\gamma_i\E_t \left[\int_0^T \hat q^i_s \, ds\right]+ \E_t\left[\int_0^T \hat\alpha^i_s\, ds\right] + \hat{B}^i_t-\sigma^i W^i_t + M^i_t \\
&= -T \bigg(h_i\eta_i + \gamma_i \tilde q_i \bigg) + \bigg( \eta_i + \psi_i  \bigg) \E_t\left [\int_0^T  P_s \,ds \right]+ \hat{B}^i_t-\sigma^i W^i_t + M^i_t \\
& = -T \bigg(   h_i\eta_i + \gamma_i \tilde q_i \bigg) + \bigg( \eta_i + \psi_i  \bigg) \left( \int_0^t  P_s \,ds  + (T-t)  P_t  \right )+ \hat{B}^i_t-\sigma^i W^i_t + M^i_t
\end{align*}
Substituting $ \E_t[\hat X^i_T]= -\frac{P_t}{2\lambda } $ in the above, and differentiating and rearranging, we get an SDE for $B^i$: 
\begin{align*}
d\hat{B}^i_t = - \left\{ \frac{1}{2\lambda} + \eta_i +\psi_i)(T-t) \right \} dP_t  + \sigma_i dW^i_t -  dM^i_t \\
\hat{B}^i_0 =-\bigg ( \frac{1}{2\lambda} + ( \eta_i + \psi_i  )T\bigg ) P_0 + T \big( h_i\eta_i + \gamma_i\tilde q_i \big)  - M^i_0,
\end{align*}
which has a unique strong solution. So, the firm has to trade in total over $[0,T]$ a number of permits equal to: 
$$ \hat B^i_T = \hat B_0^i + \int_0^T  - \left\{ \frac{1}{2\lambda} + (\eta_i +\psi_i))(T-t) \right \} dP_t  + \sigma_i W^i_T  -  M^i_T ,$$
which can be made more expressive by substituting $M^i_T=A^i_T$, the total allocation to firm $i$. Unfortunately,  an optimal control trade rate  $\hat \beta\in \L^2$   which targets the optimal cumulative number of permits $\hat B^i_T$    may fail to exist for technical reasons. In fact,  finding $\hat \beta$ amounts to represent the above random variable $\hat B_T$ as an integral with respect to the Lebesgue measure $dt$. The paper \cite{Bia-Zit} is entirely dedicated to this problem (see also \cite{bank-soner}). Their results show that a control exists if and only if the quadratic variation of  $\hat B^i$  grows \textit{slowly} when $t\rightarrow T$, or equivalently its volatility goes to zero quickly enough.  The  condition formulated in terms of volatility is precisely: 
$$\E \left[\int_0^T \frac{ (\sigma^{B^i}_t)^2}{T-t} dt   \right ] <\infty ,$$
in which $\sigma^{\hat B^i}$ is read  directly from the SDE for $\hat B^i$ above. This means that $\hat \beta^i$  exists iff the regulator's allocation moderates the volatility of price and emissions so that the exposure of $\hat B^i$   to shocks  is not too wild. If this is the case, then a canonical solution is 
$$\hat \beta^i_t = \hat B_0 +\int_0^t \frac{d\hat B_s}{T-s},$$
as shown in \cite{Bia-Zit}. This concludes the proof of the Theorem.  $\Box$

\subsection{Proof of Theorem \ref{Theo:equilibrium}}\label{proof:Theo equilibrium}
Summing over $i$  the optimal rates from \eqref{eq: opt beta}  and imposing the market clearing condition,  we get the dynamics of $\hat{P}$. The rest easily follows $\Box.$
 
\subsection{Proofs in Section 5}\label{app: Section 5}
 
\begin{Lemma}
\label{lemma: min s}
Let $a>0$ and $g:\R\to\R$ be a convex function. Then the function $s:\R\to\R$
$$ s(x) := ax^2+g(x)$$
has a unique minimizer, which is
\[
P^*:=\inf\{x \, : \ s'_+(x)\geq 0\}\in\R.
\]
\end{Lemma}

\begin{proof}
 By convexity of $g$,  there exists $m\in\R$ such that $g(x)\geq mx + g(0)$.
As a consequence,  $s$ is coercive: $\lim_{x\to\pm\infty}s(x)=+\infty$.  Being continuous, strictly convex, and coercive, $s$ has a unique minimum. Also, the right derivative $s'_+(x)=2ax+g'_+(x)$ is well- defined and non-decreasing since $s$ is convex on $\mathbb R $.  Therefore  $P^*$  above is the unique minimizer.$\square$
\end{proof}
\\

\noindent \textbf{Proof of Theorem \ref{Theo:ell-phi}}\\

\noindent For any allocation $\bA$ with corresponding equilibrium price $\hat P$ and optimal controls at equilibrium $(\hat \alpha,\hat q,\hat \beta)$,  we first isolate in the regulator's objective function a term $I^R(\mathbf{A})$ which is the sum of firms optimal cost, as per  \ref{Theo:equilibrium} and \ref{Theo:single-firm-opt}:
\begin{align*}
I^R(\bA)&= \sum_{i=1}^N \mathbb E \bigg[\int_0^T \left \{-(a_i-b_i\hat q^i_t) \hat q^i_t  +  k_i (\hat q^i_t- \tilde q_i) + \frac{\delta_i}{2} (\hat q^i_t- \tilde q_i)^2 + h_i \hat \alpha^i_t + \frac{1}{2\eta_i} (\hat{\alpha}^i_t)^2\right \} dt + \lambda (\hat X^i_T)^2 \bigg ] \\
& = \sum_{i=1}^N \bigg[\int_0^T\bigg( \kappa_i  \tilde q_i + \frac{1}{2}\delta_i ( \tilde q_i)^2 - (a_i-\kappa_i+\delta_i  \tilde q_i)\E[\hat{q}^i_t] +\bigg(b_i+\frac{\delta_i}{2}\bigg)\E[(\hat{q}^i_t)^2] +\\ 
& + h_i \E[\hat{\alpha}^i_t]+\frac{1}{2\eta_i}\E[(\hat{\alpha}^i_t)^2] \bigg) dt \bigg]  +  \frac{N}{4\lambda} \E[\hat{P}^2_T]  \\
&  = \hat P_0 T  \sum_{i=1}^N (a_i-k_i -2b_i\tilde q_i)\frac{\psi_i}{\gamma_i}+ \E\left [\int_0^T \hat P^2_t dt\right ]  \sum_{i=1}^N \Big[\Big(b_i + \frac{\delta_i}{2}\Big) \frac{\psi_i^2}{\gamma_i^2} +\eta_i\Big] + \frac{N}{4\lambda} \E[\hat{P}^2_T] + \textit{ const}\\
& = \hat P_0^2   \frac{NT}{2} \left( \bar{\phi} +\frac{1}{2\lambda T}  \right)  + N  \bar{\phi}\, \E\left [\int_0^T \langle \hat P\rangle _t dt\right ] + \frac{N}{4\lambda}  \E\left [\langle \hat P\rangle_T\right ]  + \textit{const}
\end{align*}
In the first equality, market clearing  allows us to get rid of the trade terms.  The second  and third equalities hold thanks to a combination of   \eqref{eq: opt alpha} and \eqref{eq: opt q} , together with the martingality property of $\hat{P}$, which entails $\E[\hat{P}_t]=\hat P_0$ and $\E[\hat P_t^2]=\hat P_0^2+\E[\langle \hat P\rangle_t]$ and the definition of $\tilde q_i$.  Then, the regulator goal function is:
$$ J^R(\bA) = I^R(\bA)  
 + \E[\ell( \bar{\mu}_b-\bar{\phi} \hat{P}_T+\bar H-\theta)] + \E\bigg[    \varphi \left(  \frac{\bar \omega \hat P_T}{T \pi}-\nu\right )\bigg]  $$
By Jensen's inequality and martingality of $\hat P$,   we minorize $J^R(\bA)$ by  passing the expectation sign inside the convex functions  $\ell,\varphi$ :
\begin{align*}
J^R(\bA)&\geq \hat P_0^2   N T \left(\frac{\bar \psi}{2} +\bar \eta +\frac{1}{4\lambda T}  \right)  + N  \left (\frac{\bar \psi}{2} +\bar \eta\right )\E\left [\int_0^T \langle \hat P\rangle _t\, dt\right ] + \frac{N}{4\lambda}  \E\left [\langle \hat P\rangle_T\right ]  +  \\
&   \ell\left( \bar{\mu}_b-\bar{\phi}\hat{P}_0 + \hat H-\theta \right)
+     \varphi \left(  \frac{\bar \omega \hat P_0}{T \pi}-\nu\right )  +  \textit{ const}
\end{align*}
In the right hand side now, the minimization of the quadratic variation  of $\hat{P}$ and its initial $\hat P_0$ are disentangled.  It is easy to see then that the rhs minimizer $P^*$ is unique.   It is a constant martingale, with   zero quadratic variation, and $P^*_0$    is the unique minimum (see Lemma \ref{lemma: min s}) of  the real function
\begin{align*} 
s(x) & = x^2   \frac{NT}{2} \left( \bar{\phi} +\frac{1}{2\lambda T}  \right)    + \ell\left( \bar{\mu}_b-\bar{\phi}\, x + \bar H-\theta \right)
+     \varphi \left(  \frac{\bar \omega x}{T \pi}-\nu\right ) 
\end{align*}
The optimal  firm controls follow   from Theorem \ref{Theo:equilibrium} using the above $P^*$ as equilibrium price.   
Any optimal regulation $\mathbf{A}^*$ must annihilate the volatility of the prices, so $$d\bar M^* = d\bar W,$$
The optimal average allocation $\bar M^*_0 = \E[\bar A_T] = \frac{1}{N}\E[\sum_i A^{i,*}_T]$ is found from \eqref{eq: dynamics P equilibrium}:  $$ \bar M^*_0  =(\bar H + \bar{\phi})T - {P^*} (\frac{1}{2\lambda} +\bar{\phi} T).$$
So, optimal allocations are   unique only in conditional expectation. {An example of optimal allocation is: 
$$A^{i,*}_t  = \bar M^*_0 + \sigma_i W^i_t,$$
namely the regulator gives permits in the beginning, the same number to everyone, and then neutralise the individual shocks. $\hfill\square$

\clearpage
\bibliographystyle{chicago}
\bibliography{bibliography.bib}

\end{document}